%% file: Zoom_isit.tex
\documentclass[journal]{IEEEtran}
\usepackage{amssymb,stmaryrd,amsmath,amsfonts,rotating,mathrsfs}
\usepackage[noadjust]{cite}
\usepackage{color}
\usepackage{epic}
\RequirePackage{bbm}
\input{./definition}

\allowdisplaybreaks
\begin{document}
\title{On the scaling of Polar codes:\\
I. The behavior of polarized channels }
\author{ S. Hamed Hassani, Rudiger Urbanke\thanks{EPFL, School of Computer
 \& Communication Sciences, Lausanne, CH-1015, Switzerland, \{seyedhamed.hassani,ruediger.urbanke\}@epfl.ch. This work was supported by grant no 200021-121903 of the Swiss
National Foundation.
}
}

\maketitle

\begin{abstract}
We consider the asymptotic behavior of the polarization process for
polar codes when the blocklength tends to infinity. In particular,
we study the problem of asymptotic analysis of the  cumulative distribution
$\mathbb{P}(Z_n \leq z)$, where $Z_n=Z(W_n)$ is the Bhattacharyya
process, and its dependence to the rate of transmission R. We show that for a BMS channel $W$, for $R < I(W)$ we have
$\lim_{n \to \infty} \mathbb{P} (Z_n \leq 2^{-2^{\frac{n}{2}+\sqrt{n}
\frac{Q^{-1}(\frac{R}{I(W)})}{2} +o(\sqrt{n})}}) = R$ and for $R<1- I(W)$ we have   $\lim_{n
\to \infty} \mathbb{P} (Z_n \geq 1-2^{-2^{\frac{n}{2}+  \sqrt{n}
\frac{Q^{-1}(\frac{R}{1-I(W)})}{2} +o(\sqrt{n})}}) = R$, where $Q(x)$ is the probability that the standard normal random variable will obtain a value larger than $x$. As a result, if we denote
by $\mathbb{P}_e ^{\text{SC}}(n,R)$ the probability of error using
polar codes of block-length $N=2^n$ and rate $R<I(W)$ under successive cancellation
decoding, then $\log(-\log(\mathbb{P}_e ^{\text{SC}}(n,R)))$ scales
as $\frac{n}{2}+\sqrt{n}\frac{Q^{-1}(\frac{R}{I(W)})}{2}+ o(\sqrt{n})$. We also prove that the same result holds for the block error probability using the
MAP decoder, i.e., for $\log(-\log(\mathbb{P}_e ^{\text{MAP}}(n,R)))$.
\end{abstract}

\section{Introduction}
Polar codes, recently introduced by Ar\i kan \cite{Ari09}, are a family
of codes that provably achieve the capacity of binary memoryless symmetric (BMS) channels using
low-complexity encoding and decoding algorithms. The construction of polar codes involves a method called channel polarization. In this method, $N=2^n$ copies of a BMS channel $W$ are used to construct a set of $2^n$ channels $\{W_{2^n} ^{(i)}\}_{1 \leq i \leq 2^n}$ with the property that as $n$ grows large, a fraction of almost $I(W)$ of the channels have capacity close to $1$ and a fraction of almost $1- I(W)$ of the channels have capacity close to zero. The construction of these channels is done recursively, using a transform called channel splitting. Channel splitting is a transform which takes a BMS channel $W$ as input and outputs two BMS channels $W ^+ $ and $W^-$. We denote this transform by $W \rightarrow (W^+, W^-)$.

For $N=2^n$, the construction of the channels can be visualized in the following way (\cite{Ari09}). Consider an infinite binary tree. To each vertex of the tree we assign a channel in a way that the collection of all the channels that correspond to the vertices at depth $n$ equals $\{W_{2^n} ^{(i)}\}_{1 \leq i \leq 2^n}$. We do this by a recursive procedure. Assign to the root node the channel $W$ itself. To the left offspring of the root node assign  $W^-$ and to the right one assign $W^+$. In general, if $Q$ is the channel that is assigned to vertex $v$, to the left offspring of $v$  assign  $Q ^-$ and to the right one assign $Q ^+$.
\begin{remark}
In this setting, the channel assigned to a vertex at level $n$, is obtained by starting from the original channel $W$ and applying a sequence of $+$ and $-$ on it. More precisely, label the vertices at level $n$ from left to right by $1$ to $2^n$. The channel which is assigned to the $i$-th vertex is $W_{2^n} ^{(i)}$. Let the binary representation of $i-1$ be $b_1 b_2 \cdots b_n$, where $b_1$ is the most significant bit. By the mapping $0 \to -$ and $1 \to +$, every binary sequence $b_1 b_2 \cdots b_n$ is converted to a sequence of $+$ and $-$, denoted by $c_1 c_2\cdots c_n$. Then we have
\begin{equation*}
 W_{2^n} ^{(i)}=(((W^{c_1})^{c_2})^{\cdots})^{c_n}.
\end{equation*}
\end{remark}
E.g., assuming $i=7$ we have $W_{8} ^{(7)} =((W^{+})^{+})^{-}$. For a BMS channel $W$, denote the input alphabet by $\mathcal{X}=\{0,1\}$, the output alphabet by $\mathcal{Y}$, and the transition probabilities by $W(y \mid x)$. The Bhattacharyya parameter of $W$, denoted by $Z(W)$, is given by
\begin{align*}
& Z(W)= \sum_{y \in \mathcal{Y}} \sqrt{W(y\mid 0)W(y \mid 1)}.
\end{align*}
The distribution of the Bhattacharyya parameter of the channels $\{W_ {2^n} ^{(i)}\}_{1 \leq i \leq 2^n}$ plays a fundamental role in the analysis of polar codes. More precisely, for $n \in \naturals$ and $0 < z < 1$, we are interested in analyzing the behavior of
\begin{equation}\label{equ:1}
F(n, z) = \frac{\# \{ i : Z(W_{2^n} ^{(i)}) \leq z \}}{2^n}.
\end{equation}
There is an entirely equivalent probabilistic description of
\eqref{equ:1}. Define the ``polarization''  process (\cite{ArT09})
of the channel $W$ as  $W_0=W$ and
\begin{equation}
W_{n+1}=  \left\{
\begin{array}{lr}
W_n ^{+} &  ; \text{with probability $\frac 12$},\\
W_n ^{-} &  ; \text{with probability $\frac 12$}.
\end{array} \right.
\end{equation}
In words, the process starts from the root node of the infinite binary tree and in each step moves either to the left or the right offspring of the current node with probability $\frac 12$. So at time $n$, the process $W_n$ outputs one of the $2^n$ channels at level $n$ of the tree uniformly at random. The Bhattacharyya process of the channel $W$ is defined  as $Z_n =Z(W_n)$. In this setting, we have:
\begin{equation}
 \mathbb{P}(Z_n \leq z) = F(n,z).
\end{equation}
Our objective is to investigate the behavior of  $\mathbb{P}(Z_n \leq z)$. The analysis of the  process $Z_n$ around the point $z=0$ is of particular interest since this indicates how the ``good'' channels, i.e., the channels that have mutual information close to $1$, behave. According to \cite{ArT09}, the process $Z_n$ is a super-martingale which converges almost surely to a $\{0,1\}$-valued random variable $Z_{\infty}$ with $\mathbb{P}(Z_{\infty}=0)=I(W)$. We further have,
\begin{theorem}[\cite{ArT09}] \label{AT}
 Let $W$ be a BMS channel. For any fixed $\beta < \frac 12$,
\begin{equation*}
 \liminf_{n \to \infty}\mathbb{P}(Z_n \leq 2^{-2^{n \beta}}) = I(W).
\end{equation*}
Conversely, if $I(W)<1$, then for any fixed $\beta > \frac 12$,
\begin{equation*}
 \liminf_{n \to \infty}\mathbb{P}(Z_n \geq 2^{-2^{n \beta}}) = 1.
\end{equation*}
\QED
\end{theorem}
As a result, the probability of error when using polar codes of
length $N=2^n$ under successive cancellation decoding behaves roughly
as $o(2^{-\sqrt{N}})$ as $N$ tends to infinity. Denote the error
probability by $\mathbb{P}_e ^{\text{SC}}(n,R)$. In this paper, we
provide a refined estimate of $\mathbb{P}(Z_n \leq z)$. We derive
the asymptotic relation between  $\mathbb{P}(Z_n \leq z)$ and the
rate of transmission $R$ when polar codes with a successive
cancellation decoder are used. From this we derive bounds on the
asymptotic behavior of $\mathbb{P}_e ^{\text{SC}}(n,R)$ . We further
show that the same bounds hold when we perform MAP decoding.
The outline of the paper is as follows. In Section~\ref{s_r} we state the main results of the paper. In Section~\ref{a_p} we first define several auxiliary processes and provide bounds on their asymptotic behavior. Using these bounds, we then prove the main results.
\section{Main Results}\label{s_r}
Proof of the following theorems is given in Section~\ref{a_p}.
\begin{theorem}\label{main_result}
For a BMS channel $W$, let $Z_n= Z(W_n)$ be the Bhattacharyya process of $W$.
\begin{enumerate}
\item For $R < I(W)$,
\begin{equation} \nonumber
 \lim _{n \to \infty} \mathbb{P}(Z_{n} \leq 2^{ -2^{E(n,\frac{R}{I(W)})(1+  \Theta(\frac{f (n)}{ n})) } }) =R.
\end{equation}
\item For $R < 1- I(W)$,
\begin{equation} \nonumber
 \lim _{n \to \infty} \mathbb{P}(Z_{n} \geq 1- 2^{ -2^{E(n,\frac{R}{1-I(W)}(1+  \Theta(\frac{f (n)}{ n})) } }) =R.
\end{equation}
Here, $f(n)$ is any function so that $f(n) = o(\sqrt{n})$ and $\lim_{n \to \infty} f(n) = \infty$.  The function $E(n,x)$, $0 < R< 1$,  is the unique integer solution of the equation
 \begin{equation}\label{E(n,R)}
\sum_{i=E(n,x)} ^{n} { n \choose i } \leq  2^{n} x \leq \sum_{i=E(n,x)-1} ^{n} { n \choose i }.
\end{equation}
\end{enumerate}
\QED
\end{theorem}
{\em Discussion:} Theorem~\ref{main_result} characterizes the asymptotic behavior of $\mathbb{P}(Z_n \leq z)$. By the Stirling formula applied to \eqref{E(n,R)}, the function $E(n,\frac{R}{I(W)})$ behaves like $\frac{n}{2}+\sqrt{n} \frac{Q^{-1}(\frac{R}{I(W)})}{2} +o(\sqrt{n})$, where $Q(x)$ is the probability that the standard normal random variable will obtain a value larger than $x$. Thus by Theorem~\ref{main_result} part (1) we have
$$\lim_{n \to \infty} \mathbb{P}(Z_n \leq 2^{-2^{\frac{n}{2}+\sqrt{n} \frac{Q^{-1}(\frac{R}{I(W)})}{2} +o(\sqrt{n})}}) =R.$$
This refines the result of Theorem~\ref{AT} in the following way. According to Theorem~\ref{AT}, if we transmit at rate $R$ below the channel capacity, then  $ \log(-\log(\mathbb{P}_{e} ^{\text{SC}}(n,R)))$ scales like $\frac{n}{2} + o(n)$. Theorem~\ref{main_result} gives one further term by stating that $o(n)$ is in fact $\sqrt{n}\frac{Q^{-1}(\frac{R}{I(W)})}{2}+ o(\sqrt{n})$.  The proof of Theorem~\ref{main_result} is based on observing that, once the process $Z_n$ is close to either of the endpoints of the interval $[0,1]$, it moves closer to that endpoint with high probability. As a result, the quality of a channel $W_{2^n} ^{(i)}$ is greatly dependent on the first few less significant bits of the binary expansion of $i-1$.  This observation together with the result of Theorem~\ref{main_result} imply the following.
\begin{corollary}
Let $W$ be a BMS channel and let $R < I(W)$ be the rate of transmission. The fraction of common indices chosen by polar codes and Reed-Muller codes (normalized by $2^n R$), approaches $I(W)$ as $n \to \infty$. \QED
\end{corollary}
Theorem~\ref{main_result} characterizes the scaling of the error probability of polar codes under the successive cancellation decoder. The same result holds for the case of the MAP decoder.
\begin{theorem} \label{main_result2}
Let $W$ be a BMS channel and let $R < I(W)$ be the rate of transmission. Let $C(n,R)$ be a linear code whose generator matrix is obtained by choosing a subset of $2^n R$ rows of  $\bigl[ \begin{smallmatrix} 1 & 0 \\ 1 & 1  \end{smallmatrix} \bigr]^{\otimes n}$ (e.g., polar codes or Reed-Muller codes). Denote by $\mathcal{I}_C$ the set of the indices of the chosen rows and also denote by $\mathbb{P}_C ^{\text{MAP}}(n,R)$, the block error probability when we use the code $C$ for transmission and decode according to the MAP rule. We have
\begin{equation*}
\mathbb{P}_C ^{\text{MAP}}(n,R) \geq  2^{-2^{\text{min}_{i \in \mathcal{I}_C} \text{wt($i$)} +1+\log(-\log(Z(W)))}-1},
\end{equation*}
where $\text{wt($i$)}$ denotes the number of $1$'s in the binary expansion of $i$. As a result, for every such code, we have $ \log (-\log (\mathbb{P}_C ^{\text{MAP}}(n,R))) \leq \frac{n}{2}+  \sqrt{n}\frac{ Q^{-1}(R)}{2}+ o(\sqrt{n})$. Also for the case of polar codes we have $ \log (-\log (\mathbb{P}_C ^{\text{MAP}}(n,R))) \leq \frac{n}{2} \frac{Q^{-1}(\frac{R}{I(W)})}{2} \sqrt{n} + o(\sqrt{n})$.
\QED
\end{theorem}
{\em Discussion:} By this theorem, for polar codes we have  $ \log (-\log (\mathbb{P}_C ^{\text{MAP}}(n,R))) \leq \frac{n}{2} + \sqrt{n} \frac{Q^{-1} (\frac{R}{I(W)})}{2} + o(\sqrt{n})$. Now since $\mathbb{P}_C ^{\text{MAP}}(n,R) \leq \mathbb{P}_C ^{\text{SC}}(n,R)$, for the case of polar codes $ \log (-\log (\mathbb{P}_C ^{\text{MAP}}(n,R)))$  scales as $\frac{n}{2} + \sqrt{n} \frac{Q^{-1} (\frac{R}{I(W)})}{2} + o(\sqrt{n})$.
\section{Proof of the Main Result} \label{proofs} \label{a_p}
\subsection{Analyzing closely related processes}
In this part we consider several auxiliary processes and provide bounds on their asymptotic behavior.
Let $\{B_n\}_{n \in \naturals}$ be a sequence of iid Bernoulli($\frac12$) random variables. Denote by $(\mathcal{F}, \Omega, \mathbb{P})$ the probability space generated by this sequence and let  $(\mathcal{F}_n, \Omega_n, \mathbb{P}_n)$ be the probability space generated by $(B_1, \cdots,B_n)$.  Also, denote by $\theta_n$ the natural embedding of $\mathcal{F}_n$ into  $\mathcal{F}$, i.e., for every $F \in \mathcal{F}_n$
\begin{equation*}
\theta_n(F) =\{ (b_1, b_2,\cdots, b_n, b_{n+1}, \cdots) \in \Omega \mid (b_1, \cdots, b_n ) \in F \}.
\end{equation*}
We have $\mathbb{P}_n (F) = \mathbb{P}(\theta_n (F))$.
We now couple the process $W_n$ with the sequence $\{B_i\}$:
\begin{equation}
W_{n}=  \left\{
\begin{array}{lr}
W_{n-1} ^{+} &  ; \text{if $B_n=1$},\\
W_{n-1} ^{-} &  ; \text{if $B_n=0$}.
\end{array} \right.
\end{equation}
As a result, $Z_n = Z(W_n)$ is coupled with the sequence $\{B_i\}$.
By using the bounds given in \cite[Chapter 4]{RiU08} we have the
following relationship between the Bhattacharyya  parameters of  $W
^+ $, $W^-$ and $W$:
\begin{align*} \label{general_Z}
& Z(W^+)= Z(W)^2, \\
& Z(W) \sqrt{2- Z(W)^2} \leq Z(W^-) \leq 2Z(W) - Z(W)^2.
\end{align*}
As a result, for a BMS channel $W$, the process $Z_n=Z(W_n)$ satisfies (\cite[Lemma 3.16]{Kor09thesis})
 \begin{equation} \label{general_process}
 Z_{n}  \left\{
\begin{array}{lr}
={Z_{n-1}}^2 &  ; \text{if $B_n=1$},\\
\in [Z_{n-1} \sqrt{2-{Z_{n-1}}^2} , 2Z_n-{Z_{n-1}}^2] &  ; \text{if $B_n=0$}.
\end{array} \right.
\end{equation}
Consider two processes $Z_n ^{u}$ and $Z_n ^{l}$ given by $Z_0 ^{u}=Z_0 ^{l}=Z(W)$,
\begin{equation} \label{upper_process}
{Z_n ^{u}}=  \left\{
\begin{array}{lr}
(Z_{n-1} ^{u})^2 &  ; \text{if $B_n=1$},\\
2Z_{n-1} ^{u} &  ; \text{if $B_n=0$},
\end{array} \right.
\end{equation}
and
\begin{equation} \label{lower_process}
{Z_n ^{l}}=  \left\{
\begin{array}{lr}
(Z_{n-1} ^{l})^2 &  ; \text{if $B_n=1$},\\
Z_{n-1} ^{l} &  ; \text{if $B_n=0$}.
\end{array} \right.
\end{equation}
Clearly, $Z_n$ stochastically dominates $Z_n ^{l}$ and is stochastically dominated by $Z_n ^{u}$. Also, it is easy to see that $Z_n ^{l}=(Z(W)^{2^{\sum_{i=1}^n B_i}}$. Thus
\begin{align} \label{lower_bound_map}
& \mathbb{P}(Z_n \geq (Z(W))^{2^{\sum_{i=1}^n B_i}} )\\
&= \mathbb{P}(Z_n \geq 2^{\log(Z(W))2^{\sum_{i=1}^n B_i}} ) \nonumber \\
&=1.\nonumber
\end{align}
The following lemma partially analyzes the behavior of $Z_n ^{u}$.
\begin{lemma}\label{main_lemma}
For the process $Z_n ^{u}$ (defined in \eqref{upper_process}) starting at $Z_0 ^{u}=z_0 ^{u} \in (0,1)$ we have:
\begin{equation}\label{eq:behavior}
 \mathbb{P} (Z_n ^{u} \leq 2^{- \beta 2^{\sum_{i=1} ^{n} B_i}}) \geq 1-2^{1+\frac{\beta}{2}}\sqrt{z_0^{u}}.
\end{equation}
\end{lemma}
\begin{proof}
We  analyze the process\footnote{In this paper, all the logarithms
are in base 2.} $A_n= -\log (Z_n ^{u})$ , i.e., $A_0=- \log (z_0
^{u}) \triangleq a_0$ and
\begin{equation} \label{A}
A_{n+1}=  \left\{
\begin{array}{lr}
2{A_n} &  ; \text{if } B_n=1,\\
A_n - 1   &  ; \text{if } B_n=0.
\end{array} \right.
\end{equation}
Note that in terms of the process $A_n$, the statement of the lemma can be phrased as
\begin{equation*}
 \mathbb{P} (A_n \geq \beta 2^{\sum_{i=1} ^{n} B_i}) \geq 1- \frac{2}{2^{\frac{a_0-\beta}{2}}}.
\end{equation*}
Associate to each $ (b_1, \cdots, b_n) \triangleq \omega_n  \in
\Omega_n$ a sequence of "runs" $(r_1, \cdots , r_{k(\omega_n)})$. This sequence
is constructed by the following procedure. We define $r_1$ as the smallest index
$i \in \naturals$ so that $b_{i+1} \neq b_1$. In general, if $\sum_{j=1}
^{k-1} r_j < n$ then
\begin{align*}
r_k = \min \{i \mid  \sum_{j=1} ^{k-1} r_j < i \leq n , b_{i+1} \neq b_{\sum_{j=1} ^{k-1} r_j}\}-\sum_{j=1} ^{k-1} r_j.
\end{align*}
The process stops whenever the sum of the runs equals $n$. Denote the
stopping time of the process by $k(\omega_n)$. In words, the sequence
$(b_1,\cdots, b_n)$ starts with $b_1$. It then repeats $b_1$, $r_1$ times.
Next follow $r_2$ instances of $\overline{b_1}$, followed again by $r_3$
instances of $b_1$, and so on. We see that $b_1$ and $(r_1, \cdots,
r_{k(\omega_n)})$ fully describe $\omega_n =(b_1,\cdots, b_n) $. Therefore,
there is a one-to-one map
\begin{equation} \label{one-to-one}
 (b_1,\cdots, b_n) \longleftrightarrow \{ b_1, (r_1, \cdots, r_{k(\omega_n)})\}.
\end{equation}
Note that we can either have $b_1 = 1$ or $b_1 = 0$. We start with the first case, i.e., we first assume $B_1=1$. We have:
\begin{equation*}
\sum_{i=1}^{n} b_i = \sum_{ \text{$j$ odd $\leq k(\omega_n)$}} r_j,
\end{equation*}
and
\begin{equation*}
 n = \sum_{j=1} ^{k(\omega_n)} r_j.
\end{equation*}
Analogously, for a realization $(b_1,b_2, \cdots) \triangleq \omega \in \Omega$ of the infinite sequence of random variable $\{B_i\} _{i \in \naturals}$, we can associate a sequence of runs $(r_1,r_2, \cdots)$. In this regard, considering the infinite sequence of random variables $\{B_i\} _{i \in \naturals}$ (with the extra condition $B_1 =1$), the corresponding sequence of runs, which we denote by $\{ R_k \}_{k \in \naturals}$, is an iid sequence with $\mathbb{P}(R_i= j) = \frac{1}{2^j}$.
Let us now see how we can express the $A_n$ in terms of the $r_1,
r_2, \cdots, r_{k(\omega_n)}$.  We begin by a simple example: Consider
the sequence $(b_1=1, b_2, \cdots,b_8)$ and the associated run sequence
$(r_1,\cdots, r_5)=(1,2,1,3,1)$. We have
\begin{align*}
A_{1}    &= a_0 2^{r_1} , \\
A_{3}    &= a_0 2^{r_1}  - r_2, \\
A_{4}    &= (a_0 2^{r_1}  - r_2) 2^{r_3}= a_0 2^{r_1+r_3}  - r_2 2^{r_3}  , \\
A_{7}    &= (a_0 2^{r_1}  - r_2) 2^{r_3} - r_4= a_0 2^{r_1+r_3}  - r_2 2^{r_3} - r_4  ,\\
A_{8}    &= ((a_0 \times 2^{r_1} - r_2) \times 2^{r_3} -r_4) \times 2^{r_5} \\
         & = a_0 2^{r_1+r_3+r_5} -r_2 2^{r_3+r_5} - r_4 2^{r_5} \\
         & = 2^{r_1+r_3+r_5}( a_0 - 2^{-r_1} r_2- 2^{-(r_1+r_3)} r_4).
\end{align*}
In general, for a sequence $(b_1, \cdots, b_n)$ with the associated run
sequence $(r_1,\cdots, r_{k(\omega_n)})$ we can write:
\begin{align*}
A_{n} & = a_0 2^{\sum_{\text{$i$ odd $\leq k(\omega_n)$}} r_i} - \!  \!  \!  \!  \!  \!  \! \sum_{\text{$i$ even $\leq k(\omega_n)$}} \!  \!  \!  \!  \!  \!  \! r_i  2 ^{\sum_{\text{$i < j$ odd }} r_j} \\
&= a_0 2^{\sum_{\text{$i$ odd $\leq k(\omega_n)$}} r_i} -\!  \!  \!  \!  \!  \!  \!  \!  \!  \!  \!  \!  \sum_{\!  \text{$i$ even $\leq k(\omega_n)$}} \!  \!  \!  \!  \!  \!  \!  \!   r_i  2 ^{(-\sum_{\text{$j$ odd $< i$}} r_j + \sum_{\text{$i$ odd $\leq k(\omega_n)$}} r_i )} \\
&= [2^{\sum_{\text{$i$ odd $\leq k(\omega_n)$}} r_i}][a_0 -  (\!  \!  \!  \!  \!  \!  \! \sum_{\text{$i$ even $\leq k(\omega_n)$}} \!  \!  \!  \!  \!  \!  \! r_i  2 ^{-\sum_{\text{$j$ odd $< i$}} r_j} ) ] \\
&= [2^{\sum_{i=1} ^{n} B_i}] [a_0 -  (\!  \!  \!  \!  \!  \!  \! \sum_{\text{$i$ even $\leq k(\omega_n)$}} \!  \!  \!  \!  \!  \!  \! r_i  2 ^{-\sum_{\text{$j$ odd $< i$}} r_j }) ].
\end{align*}
Our aim is to lower-bound
\begin{align*}
& \mathbb{P} (A_n \geq \beta 2^{\sum_{i=1} ^{n} B_i})\\
&= \mathbb{P}_n  (a_0 -  \sum_{\text{$i$ even $\leq k(\omega_n)$}} r_i  2 ^{-\sum_{\text{$j$ odd $< i$}} r_j} \geq \beta ),
\end{align*}
or, equivalently, to upper-bound
\begin{equation}  \label{prob}
\mathbb{P}_n (\sum_{\text{$i$ even $\leq k(\omega_n)$}} r_i  2 ^{-\sum_{\text{$j$ odd $< i$}} r_j} \geq a_0-\beta).
\end{equation}
For $n \in \naturals$, define the set $U_n \in \mathcal{F}_n$ as
\begin{equation*}
U_n =  \{ \omega_n \in \Omega_n \mid  \exists l \leq k(\omega_n) : \sum_{\text{$i$ even $\leq l$}} r_i  2 ^{-\sum_{\text{$j$ odd $< i$}} r_j} \geq a_0-\beta   \}.
\end{equation*}
Clearly we have:
\begin{equation*}
\mathbb{P}_n  (\sum_{\text{$i$ even $\leq k(\omega_n)$}} r_i  2 ^{-\sum_{\text{$j$ odd $< i$}} r_j} \geq a_0-\beta ) \leq \mathbb{P}_n(U_n).
\end{equation*}
In the following we show that if   $(b_1, \cdots,b_n) \in U_{n} $, then for any choice of $b_{n+1}$, $(b_1, \cdots,b_n,b_{n+1}) \in U_{n+1}$. We will only consider the case when $b_n,b_{n+1}=1$, the other three cases can be verified similarly. Let $\omega_n = (b_1, \cdots, b_{n-1},b_n=1) \in U_n$. Hence, $k(\omega_n)$ is an odd number (recall that $b_1=1$) and the quantity $\sum_{\text{$i$ even $\leq k(\omega_n)$}} r_i  2 ^{-\sum_{\text{$j$ odd $< i$}} r_j}$ does not depend on $r_{k(\omega_n)}$. Now consider the sequence  $\omega_{n+1}=(b_1, \cdots,b_n=1, 1)$. Since the last bit ($b_{n+1}$) equals $1$, then $r_{k(\omega_{n+1})}=r_{k(\omega_n)}$ and the value of the sum remains unchanged. As a result $(b_1, \cdots,b_n,1) \in U_{n+1}$.
From above, we conclude that $\theta_i(U_i) \subseteq \theta_{i+1}(U_{i+1})$ and as a result
\begin{align*}
\mathbb{P}_i (U_i) &= \mathbb{P}(\theta_i(U_i)) \leq \mathbb{P}(\theta_{i+1}(U_{i+1})) =\mathbb{P}_{i+1}(U_{i+1}).
\end{align*}
Hence, the quantity $\lim_{n \to \infty} \mathbb{P}_n (U_n) =\lim_{n \to \infty} \mathbb{P}  (\theta_n(U_n))=\lim_{n \to \infty} \mathbb{P}(\cup_{i=1}^{n} \theta_i (U_i))$ is an upper bound on \eqref{prob}.
On the other hand, consider the set
\begin{equation*}
V =  \{ \omega \in \Omega \mid \exists l  : \sum_{\text{$i$ even $\leq l$}} r_i  2 ^{-\sum_{\text{$j$ odd $< i$}} r_j} \geq a_0-\beta \}.
\end{equation*}
By the definition of  $V$ we  have $ \cup_{i=1} ^{\infty} \theta_i (U_i) \subseteq V$, and as a result,  $\mathbb{P} (\cup_{i=1} ^{\infty} \theta_i (U_i)) \leq \mathbb{P}(V)$.   In order to bound the probability of the set $V$, note that assuming $B_1=1$,  the sequence $\{ R_k \}_{k \in \naturals}$ (i.e., the sequence of runs when associated with the sequence $\{B_i\}_{i \in \naturals}$) is an iid sequence with $\mathbb{P}(R_i= j) = \frac{1}{2^j}$.  We also have
\begin{align} \label{markov}
& \mathbb{P} (a_0 -  \sum_{\text{$i$ even $\leq m$}} R_i  2 ^{-\sum_{\text{$j$ odd $ < i$}} R_j} \leq \beta )\\
& = \mathbb{P}  (\sum_{\text{$i$ even $\leq m$}} R_i  2 ^{-\sum_{\text{$j$ odd $ < i$}} R_j} \geq  a_0 -\beta) \nonumber \\
& = \mathbb{P}  (2^{\sum_{\text{$i$ even $\leq m$}} R_i  2 ^{-\sum_{\text{$j$ odd $ < i$}} R_j}} \geq  2^{a_0 -\beta})\nonumber \\ 
& \leq \frac{\mathbb{E}[2^{\sum_{\text{$i$ even $\leq m$}} R_i  2 ^{-\sum_{\text{$j$ odd $ < i$}} R_j }}]}{2^{a_0-\beta}}, \nonumber
\end{align}
where the last step follows from the Markov inequality. The idea is now to provide an upper bound on the quantity $\mathbb{E}[2^{\sum_{\text{$i$ even $\leq m$}} R_i  2 ^{-\sum_{\text{$j$ odd $ < i$}} R_j }}]$. Let $X=\sum_{\text{$i$ even $\leq m$}} R_i  2 ^{-\sum_{\text{$j$ odd $ < i$}} R_j }$. We have
\begin{align*}
& \mathbb{E}[2^X] \\
& =\sum_{l=1}^{\infty} \mathbb{P}(R_2 = l) \mathbb{E}[2^X \mid R_2= l]\\
& \stackrel{a}{=}\sum_{l=1}^{\infty} \frac{1}{2^{l}} \mathbb{E}[2^X \mid R_2= l]\\
&=  \sum_{l=1}^{\infty} \frac{1}{2^{l}}  \mathbb{E}[2^{\frac{R_1}{2^{l}}}]\mathbb{E}[2^{\frac{X}{2^{l}}}]\\
& = \sum_{l=1}^{\infty} \frac{1}{2^{l}( 2^{1-\frac{1}{2^l}})} \mathbb{E}[2^{\frac{X}{2^{l}}}]\\
& \stackrel{b}{\leq} \sum_{l=1}^{\infty} \frac{1}{2^{l}( 2^{1-\frac{1}{2^l}})}  (\mathbb{E}[2^X])^{\frac{1}{2^l}},\\
\end{align*}
where (a) follows from the fact that $R_i$s are iid and $X$ is self-similar and (b) follows from Jensen inequality. As a result , an upper bound on the quantity $\mathbb{E}[2^X]$ can be derived as follows. We have
\begin{small}
\begin{align*}
& \mathbb{E}[2^X  \! \! \! ] \leq  \! \! \!  \frac{1}{2(2^{\frac 12}-1 )} ( \mathbb{E}[2^X])^{\frac{1}{2}}  \! \! \!  +  \! \! \!  \frac{1}{4(2^{\frac 34}-1 )}  (\mathbb{E}[2^X])^{\frac{1}{4}}  \! \! \!  +  \! \! \!  \frac{1}{4(2^{\frac 78}-1 )}  (\mathbb{E}[2^X])^{\frac{1}{8}}.
\end{align*}
\end{small}
The equation $ y = \frac{1}{2(2^{\frac 12}-1 )}  y^{\frac{1}{2}} + \frac{1}{4(2^{\frac 34}-1 )}  y^{\frac{1}{4}} + \frac{1}{4(2^{\frac 78}-1 )}  y^{\frac{1}{8}}$ has only one real valued solution $y^* \leq  2.87$. As a result we have $\mathbb{E}[2^X] \leq y^* \leq 2.87$. Thus by \eqref{markov} we obtain 
\begin{align*}
\mathbb{P} (a_0 -  \sum_{\text{$i$ even $\leq m$}} R_i  2 ^{-\sum_{\text{$j$ odd $ < i$}} R_j} \leq \beta ) \leq \frac{2.87}{2^{a_0 - \beta}}
\end{align*}
Thus, given that $B_1=1$, we have:
\begin{equation*}
 \mathbb{P} (A_n \geq \beta 2^{\sum_{i=1} ^{n} B_i}) \geq 1- \frac{2.87}{2^{a_0-\beta}}.
\end{equation*}
Or more precisely we have
\begin{equation*}
 \mathbb{P} (A_n \geq \beta 2^{\sum_{i=1} ^{n} B_i} \mid B_1=1) \geq 1- \frac{2.87}{2^{a_0-\beta}}.
\end{equation*}
 Now consider the case $B_1=0$. We show that a similar bound applies for $A_n$. Firstly note that, fixing the value of $n$, the distribution of $R_1$ is as follows: $\mathbb{P}(R_i)=\frac{1}{2^i}$ for $1 \leq i \leq n-1$ and $\mathbb{P}(R_1=n) = \frac{1}{2^{n-1}}$. We have
\begin{small}
\begin{align*}
 & \mathbb{P} (A_n \geq \beta 2^{\sum_{i=1} ^{n} B_i} \mid B_1=0)\\
& = \sum_{i=1}^{n} \mathbb{P} (A_n \geq \beta 2^{\sum_{i=1} ^{n} B_i} \mid R_1 = i, B_1=0) \mathbb{P}(R_1 =i \mid B_1=0)\\
& =  \! \! \!  \! \! \!  \! \! \!   \sum_{i \leq a_0 -\beta,  i \leq n}  \! \! \!  \! \! \!  \! \! \!   \mathbb{P} (A_n \geq \beta 2^{\sum_{i=1} ^{n} B_i}\mid R_1 = i, B_1=0)\mathbb{P}(R_1 =i \mid B_1=0)\\
& + \sum_{i > a_0 -\beta, i \leq n } ^{n} \mathbb{P}(R_1=i \mid B_1 =0)\\
&\leq  \sum_{i \leq a_0 -\beta,  i \leq n} \frac{1}{2^i} \frac{2.87}{2^{a_0 - \beta -i}} + \frac{2}{2^{a_0-\beta}}\\
&\leq \frac{2.87(a_0-\beta+1)}{2^{a_0-\beta}}\\
&\leq \frac{3}{2^{\frac{a_0-\beta}{2}}}.
\end{align*}
\end{small}
Hence, considering the two cases together, we have:
\begin{equation*}
 \mathbb{P} (A_n \geq \beta 2^{\sum_{i=1}^{n} B_i}) \geq 1- \frac{2}{2^{\frac{a_0-\beta}{2}}}.
\end{equation*}
\end{proof}
As a result of the above lemma, if the initial point of the process $Z_n^u$ is sufficiently close to zero, its behavior is close to the behavior of the process $Z_n^l$. The same phenomenon occurs for the process $Z_n$ since it is sandwiched between $Z_n^l$ and $Z_n^u$. The following statement relates the behavior of the processes $Z_n^u$ and $Z_n^l$.
\begin{corollary} \label{E(n,x)}
Let $Z_n^u$ be the process given in \eqref{upper_process} with $Z_0^u=z_0^u \in (0,1)$. For $x \in (0,1)$ we have
\begin{equation}
\mathbb{P}(Z_n ^{u} \leq 2^{-2^{E(n,x)}}) \geq x-2\sqrt{2} \sqrt{z_0^u}-o(\frac{1}{\sqrt{n}}).
\end{equation}
\end{corollary}
\begin{proof}
Recall  $E(n,x)$ from \eqref{E(n,x)} and let the two events $A$ and $B$ be defined as follows,
\begin{align*}
&A=\{(b_1, \cdots, b_n)\in \Omega_n \mid Z_n ^{u} (b_1,\cdots,b_n) \leq 2^{- 2^{\sum_{i=1} ^{n} b_i}} \},\\
&B=\{(b_1, \cdots, b_n)\in \Omega_n \mid 2^{-2^{\sum_{i=1} ^{n} b_i}} \leq 2^{- 2^{E(n,x)}} \}.
\end{align*}
By inserting $\beta=1$ in Lemma~\ref{main_lemma} we obtain $\mathbb{P}(A) \geq 1- 2 \sqrt{2}\sqrt{z_0^u}$ and
\begin{align*}
\mathbb{P}(B)&=\mathbb{P}(\sum_{i=1}^n B_i \geq E(n,x))\\
&\geq x-o(\frac{1}{\sqrt{n}}).
\end{align*}
As a result,
 \begin{align*}
 &\mathbb{P}(Z_n \leq 2^{-2^{E(n,x)}}) \\
 &\geq \mathbb{P}(A \cap B)\\
 &= \mathbb{P}(A) +\mathbb{P}(B)-\mathbb{P}(A \cup B)\\
 &\geq \mathbb{P}(A)+\mathbb{P}(B)-1\\
 &\geq x-2\sqrt{2}\sqrt{z_0^u}-o(\frac{1}{\sqrt{n}}).
 \end{align*}
\end{proof}
\subsection {Proof of Theorem~\ref{main_result}}
We start with the proof of part (1). The main idea behind the proof is to analyze the behavior of the process $Z_n$ once its value is sufficiently close to the endpoints of the interval. In this regard, we first give a bound on the speed of converging to the endpoints. The proof of following lemma is given in the appendix.
\begin{lemma} \label{gen_bound}
Let $W$ be a BMS channel and $Z_n=Z(W_n)$ be the corresponding Bhattacharyya process.  Let $\rho \in ((\frac{1.85}{2})^{\frac 23},1)$ be a fixed constant. There exist constants $\alpha_1,\alpha_2 \geq 0$, independent on $\rho$, such that
\begin{enumerate}
\item[(a)] $\mathbb{P}(Z_n \leq 2{\rho}^{n})\geq  I(W)-\alpha_1 {\rho}^{\frac{n}{2}}$.
\item[(b)] $\mathbb{P}(Z_n \geq 1 - 2{\rho}^{n})\geq 1-I(W)-\alpha_2 {\rho}^n$.
\end{enumerate}
\QED
\end{lemma}
We then proceed by by providing upper and lower bounds on the quantity
 \begin{align*}
\mathbb{P}(Z_{n} \leq 2^{ - 2^{E(n,x)(1+ \Theta(\frac{f(n)}{E(n,x)}))}}),
\end{align*}
and by showing that as $n$ grows large, both of the bounds tend to R.
\subsubsection{Lower bound}
Fix $m \in \naturals$ and let $x =\frac{R}{I(W)}$. By  Lemma~\ref{gen_bound}, we have:
\begin{align*}
 \mathbb{P} (Z_{m} \leq 2 {\rho}^m) \geq I(W) -  \alpha_1 {\rho}^{\frac{m}{2}}.
\end{align*}
As a result,
\begin{align*}
&\mathbb{P}(Z_{n+m} \leq 2^{- 2^{E(n,x)}} )\\
&\geq \mathbb{P}(Z_{n+m} \leq 2^{- 2^{E(n,x)}}  \mid Z_{m} \leq 2 {\rho}^m) \mathbb{P}(Z_{m} \leq 2 {\rho}^m)\\
 &\geq \mathbb{P}(Z_{n+m} \leq 2^{- 2^{E(n,x)}}  \mid Z_{m} \leq 2 {\rho}^m)(I(W)-\alpha_1 {\rho}^{\frac{m}{2}})\\
 & \geq  (x-2\sqrt{2}{\rho}^{\frac{m}{2}} -o(\frac{1}{\sqrt{n}}))(I(W)-\alpha_1 {\rho}^{\frac{m}{2}}),
\end{align*}
where the last inequality follows from Corollary~\ref{E(n,x)} and the fact that assuming $Z_{m} \leq 2 {\rho}^m$, the process $Z_{n+m}$ is dominated by the process $Z_{n} ^u$ with the initial condition $Z_0^u=z_0^u=2 {\rho}^m$.
Now, since $E(n+m,x)-m \leq E(n,x)$ and $xI(W)=R$,  we have
\begin{align}\label{bound6}
&\mathbb{P}(Z_{n+m} \leq 2^{ - 2^{E(n+m,x) -m}} ) \\\nonumber
&\geq  R-2\sqrt{2} \alpha_1 {\rho}^{{m}} -  (2\sqrt{2}I(W)+x\alpha_1){\rho}^{\frac{m}{2}}-o(\frac{1}{\sqrt{n}}). \nonumber
\end{align}
Thus by changing the variable $n\leftarrow n+m$, for every $m,n \in \naturals$ such that $n \geq m$ we have
\begin{align}\label{bound2}
&\mathbb{P}(Z_{n} \leq 2^{ - 2^{E(n,x) -m}} ) \\\nonumber
&\geq  R-2\sqrt{2} \alpha_1 {\rho}^{{m}} - (2\sqrt{2}I(W)+x \alpha_1){\rho}^{\frac{m}{2}}-o(\frac{1}{\sqrt{n-m}}). \nonumber
\end{align}
\subsubsection{Upper bound}
Consider $m$ and $x$ as above. By \eqref{lower_bound_map} we have:
\begin{equation*}
\mathbb{P}(Z_{n+m} \geq (Z(W))^{2^{m}2^{\sum_{i=m+1}^{n+m}B_i}}) =1.
\end{equation*}
As a result,
\begin{equation*}
\mathbb{P}(Z_{n+m} \geq (Z(W))^{2^{m}2^{\sum_{i=m+1}^{n+m}B_i}} \mid Z_{m} \leq 2 {\rho}^m) =1.
\end{equation*}
Therefore,
\begin{small}
\begin{align} \label{bound5}
& \mathbb{P}(Z_{n+m} \geq (Z(W))^{2^{m}2^{E(n,x)}} \mid Z_{m} \leq 2 {\rho}^m) \\ \nonumber
& \geq \mathbb{P}((Z(W))^{2^{m}2^{\sum_{i=m+1}^{n+m}B_i}} \geq (Z(W))^{2^{m}2^{E(n,x)}} \mid Z_{m} \leq 2 {\rho}^m)\\ \nonumber
&=\mathbb{P}(\sum_{i=m+1}^{n+m}B_i \leq E(n,x) \mid Z_{m} \leq 2 {\rho}^m )\\ \nonumber
&=\mathbb{P}(\sum_{i=m+1}^{n+m}B_i \leq E(n,x)) \\ \nonumber
& \geq 1-x-o(\frac{1}{\sqrt{n}}),  \nonumber
\end{align}
\end{small}
and
\begin{small}
\begin{align*}
& \mathbb{P}(Z_{n+m} \geq (Z(W))^{2^{m}2^{E(n,x)}},Z_{m} \leq 2 {\rho}^m) \\
&= \mathbb{P}(Z_{n+m} \geq (Z(W))^{2^{m}2^{E(n,x)}} \mid Z_{m} \leq 2 {\rho}^m) \mathbb{P}(Z_{m} \leq 2 {\rho}^m) \\
&\geq (1-x-o(\frac{1}{\sqrt{n}}))(I(W)- \alpha_1 {\rho}^{\frac{m}{2}}).
\end{align*}
\end{small}
As a result, we have
\begin{small}
 \begin{align}\label{bound1}
& \mathbb{P}(Z_{n+m} \leq (Z(W))^{2^{m}2^{E(n,x)}},Z_{m} \leq 2 {\rho}^m) \\\nonumber
&=\mathbb{P}(Z_{m} \leq 2 {\rho}^m)-\mathbb{P}(Z_{n+m} \geq (Z(W))^{2^{m}{E(n,x)}},Z_{m} \leq 2 {\rho}^m) \\\nonumber
&\leq 1-(1-I(W)-\alpha_2 \rho^{m})-(1-x-o(\frac{1}{\sqrt{n}}))(I(W)-\alpha_1 {\rho}^{\frac{m}{2}})\\\nonumber
&\leq I(W)-(1-x)I(W)+ \alpha_2 \rho^{m}+ \alpha_1 {\rho}^{\frac{m}{2}}+o(\frac{1}{\sqrt{n}})\\\nonumber
&=xI(W)+ \alpha_2 \rho^{m}+ \alpha_1 {\rho}^{\frac{m}{2}}+o(\frac{1}{\sqrt{n}}).\nonumber
\end{align}
\end{small}
Also note that
 \begin{align}\label{bound}
& \mathbb{P}(Z_{n+m} \leq (Z(W))^{2^{m}2^{E(n,x)}})\\\nonumber
&=\mathbb{P}(Z_{n+m} \leq (Z(W))^{2^{m}2^{E(n,x)}},Z_{m} \leq 2 {\rho}^m)\\ \nonumber
&+\mathbb{P}(Z_{n+m} \leq (Z(W))^{2^{m}2^{E(n,x)}},Z_{m} \geq 2 {\rho}^m). \nonumber
\end{align}
We now upper bound the quantity $\mathbb{P}(Z_{n+m} \leq (Z(W))^{2^{m}2^{E(n,x)}},Z_{m} \geq 2 {\rho}^m)$. Firstly note that as $m$ grows large we have $(Z(W))^{2^{m}2^{E(n,x)}} \leq 2 {\rho}^m$. More precisely if we choose $m$ large enough so that  the inequality
\begin{equation}\label{m}
2^m \geq m \frac{\log \rho}{\log Z(W)},
\end{equation}
is fulfilled, then the relation $(Z(W))^{2^{m}2^{E(n,x)}} \leq 2 {\rho}^m$ holds. For this choice of $m$ we have
\begin{small}
\begin{align*}
&\mathbb{P}(Z_{n+m} \leq (Z(W))^{2^{m}2^{E(n,x)}},Z_{m} \geq 2 {\rho}^m)\\
&\leq \mathbb{P}(Z_{n+m} \leq 2 {\rho}^m, Z_m \geq 2 {\rho}^m )\\
& = \mathbb{P}(Z_{n+m} \leq 2 {\rho}^m  \! \! \!  ,2 {\rho}^m  \! \! \!  \leq Z_m \leq 1-2 {\rho}^m )\mathbb{P}(2 {\rho}^m \leq  \! \!  Z_m \leq  \! \!  1-2 {\rho}^m )\\
& \quad + \mathbb{P}(Z_{n+m} \leq 2 {\rho}^m \mid Z_m \geq 1-2 {\rho}^m )\mathbb{P}(Z_m \geq 1-2 {\rho}^m )\\
&\leq\mathbb{P}(2 {\rho}^m \leq Z_m \leq 1-2 {\rho}^m ) + \mathbb{P}(Z_{n+m} \leq 2 {\rho}^m \mid Z_m \geq 1-2 {\rho}^m ).
\end{align*}
\end{small}
Now, by Lemma~\ref{gen_bound} it is easy to see that
\begin{small}
\begin{align*}
\mathbb{P}(2 {\rho}^m \leq Z_m \leq 1-2 {\rho}^m )& \leq 1-(I(W)- \alpha_1 {\rho}^{\frac{m}{2}})\\
&-(1-I(W)-\alpha_2 \rho^{m})\\
&=  \alpha_1 {\rho}^{\frac{m}{2}}+\alpha_2 \rho^{m}.
\end{align*}
\end{small}
Also to upperbound $\mathbb{P}(Z_{n+m} \leq 2 {\rho}^m \mid Z_m \geq 1-2 {\rho}^m )$, note that if we consider the process $E_n$ given in \eqref{E} with the initial condition $e_0 = 1- 2 {\rho}^m$, then as a result of Lemma~\ref{lower_gen} we have
\begin{align*}
&\mathbb{P}(Z_{n+m} \leq 2 {\rho}^m \mid Z_m \geq 1-2 {\rho}^m )\\
&\leq \mathbb{P}(E_n \leq 2 {\rho}^m)\\
&\leq 2\sqrt{2}\sqrt{1-(1-2 {\rho}^m)^2}\\
&\leq 8{\rho}^{\frac{m}{2}}.
\end{align*}
Summing up the above arguments, we have
\begin{align}
 \mathbb{P}(Z_{n+m} \leq 2 {\rho}^m, Z_m \geq 2 {\rho}^m ) \leq   (\alpha_1+ 8) {\rho}^{\frac{m}{2}}+ \alpha_2 \rho^{m}.
\end{align}
And as a result, for $m$ large enough so that \eqref{m} is fulfilled we have
\begin{align*}
&\mathbb{P}(Z_{n+m} \leq (Z(W))^{2^{m}2^{E(n,x)}},Z_{m} \geq 2 {\rho}^m)\\
&\leq  (\alpha_1+ 8) {\rho}^{\frac{m}{2}}+ \alpha_2 \rho^{m}.
\end{align*}
Plugging this into \eqref{bound} and using \eqref{bound1}, we have
 \begin{align*}
&\mathbb{P}(Z_{n+m} \leq (Z(W))^{2^{m}2^{E(n,x)}}) \\ \nonumber
&\leq xI(W)+ 2\alpha_2 \rho^{m}+ (2 \alpha_1+ 8) {\rho}^{\frac{m}{2}}+o(\frac{1}{\sqrt{n}}). \nonumber
 \end{align*}
Also, since $E(n,x) \leq E(n+m,x)$ and $xI(W)=R$ we have:
  \begin{align*}
&\mathbb{P}(Z_{n+m} \leq (Z(W))^{2^{m}2^{E(n+m,x)}}) \\ \nonumber
&\leq xI(W)+ 2\alpha_2 \rho^{m}+ (2 \alpha_1+ 8) {\rho}^{\frac{m}{2}}+o(\frac{1}{\sqrt{n}}). \nonumber
 \end{align*}
Thus by changing the variable $n\leftarrow n+m$, for every $m,n \in \naturals$ such that $n \geq m$ we have
   \begin{align}\label{bound3}
&\mathbb{P}(Z_{n} \leq (Z(W))^{2^{m}2^{E(n,x)}}) \\ \nonumber
&\leq R+ 2\alpha_2 \rho^{m}+ (2 \alpha_1+ 8) {\rho}^{\frac{m}{2}}+o(\frac{1}{\sqrt{n-m}}) . \nonumber
 \end{align}
 \subsubsection{Combining the upper and lower bounds}
Recall that $f(n)$ is any function so that $f(n) = o(\sqrt{n})$ and $\lim_{n \to \infty} f(n) = \infty$. Thus by letting $m=f(n)$, as $n$ grows large, we have $m \ll n$ and by using \eqref{bound2} we have
 \begin{align*}
& \lim_{n \to \infty} \mathbb{P}(Z_{n} \leq 2^{ - 2^{E(n,x)(1+ \frac{-f(n)}{E(n,x)})}}) \geq R. \\\nonumber
\end{align*}
Therefore,
 \begin{align*}
& \lim_{n \to \infty} \mathbb{P}(Z_{n} \leq 2^{ - 2^{E(n,x)(1+ \Theta(\frac{f(n)}{E(n,x)})}}) \geq R. \\\nonumber
\end{align*}
Also, as $\lim_{n \to \infty}f(n) =\infty$, inequality \eqref{m} is fulfilled as $n$ grows large, and by \eqref{bound3}, we have
   \begin{align*}
\lim_{n \to \infty}\mathbb{P}(Z_{n} \leq 2^{-2^{E(n,x)+f(n)+\log(-\log(Z(W)))}}) \leq R.
 \end{align*}
And as a result,
   \begin{align*}
\lim_{n \to \infty}\mathbb{P}(Z_{n} \leq 2^{-2^{E(n,x)(1+\Theta(\frac{f(n)}{E(n,x)}))}}) \leq R.
 \end{align*}
Therefore, since the limit of the upper and lower bound equals $R$, we get the result.

To prove part (2), we first consider the process $Z'_n= 1-Z_n^2$. By using \eqref{general_process} we have
\begin{small}
\begin{equation*} \left\{
\begin{array}{lr}
Z'_{n+1}=1-Z_{n+1} ^2 \leq 1-Z_{n} ^4 \leq 2(1-Z_{n} ^2 )=2 Z'_n    &  ; \text{if } \bar{B}_n=1,\\
Z'_{n+1}=1-Z_{n+1} ^2 \leq (1-Z_{n} ^2)^2 = {Z'_n}^2  &  ; \text{if } \bar{B}_n=0.
\end{array} \right.
\end{equation*}
\end{small}
Thus the process $Z'_n$ with is stochastically dominated by the process $Z_n^u$ given by \eqref{upper_process} with $Z_0^u=Z'_0$. Also by using Lemma~\ref{gen_bound}, for $m \in \naturals$ we have
\begin{align*}
& \mathbb{P}(Z'_m \leq 2 \rho^m) \\
& = \mathbb{P}(1-Z'_m\geq 1-2 \rho^m) \\
&  =\mathbb{P}(Z_m^2 \geq 1-2 \rho^m) \\
&  =\mathbb{P}(Z_m \geq \sqrt{1-2 \rho^m)} \\
& \geq \mathbb{P}(Z_m \geq 1- \rho^m) \\
&\geq1- I(W) - 2\alpha_2 \rho^{m}.
\end{align*}
Similarly we obtain 
\begin{align*}
\mathbb{P}(Z'_m \geq 1-2 \rho^m) \geq I(W)- \alpha_1 {\rho}^{\frac{m}{2}}.
\end{align*}
Using the above statements for the process $Z'_n$ and going along the same lines as the proof of part (1), for $R < 1-I(W)$  we obtain
\begin{align*}
\lim_{n \to \infty}\mathbb{P} (Z'_n \leq 2^{-2^{E(n,\frac{R}{1-I(W)}) (1+\Theta(\frac{f(n)}{n}))} })=R,
\end{align*}
and by noting that $Z'_n = 1-Z_n^2$ we get the result.
\subsection{Proof of Theorem~\ref{main_result2}}
Let $\mathcal{I}$ be the set of chosen indices by the code $C(n,R)$,
let $U_1 ^{2^n}$ the block to be transmitted (including the
frozen bits), and let $Y_1 ^{2^n }$ be the received vector. Denote
by $\mathbb{P}_{e,i} ^{\text{MAP}}(N,R)$  the bit-error probability
when we decode the $i$-th bit by the MAP rule.  We have
\begin{align*}
\mathbb{P}_e ^{\text{MAP}} (N,R) & \stackrel{(a)}{\geq} \text{max}_{i \in \mathcal{I} } \{ \mathbb{P}_{e,i} ^{\text{MAP}} (N,R) \}\\
&\stackrel{(b)}{\geq}  \text{max}_{i \in \mathcal{I} } \{ H(U_i \mid Y_{1} ^{2^n } )\}  \\
& \geq \text{max}_{i \in \mathcal{I} } \{ H(U_i \mid Y_{1} ^{2^n }, U_1 ^{i-1}, U_{i+1}^{2^n})\}\\
&=\text{max}_{i \in \mathcal{I} } \{ H(\bar{W}_i) \},
\end{align*}
where $\bar{W}_i$ is the channel seen by $U_i$ when we have the
output $Y_{1} ^{2^n}$ and all the other bits $U_1,\cdots,U_{i-1},
U_{i+1}, \cdots, U_{2^n}$ available. To see step (a) consider the MAP decoder
for bit $i$. It has associated probability $\mathbb{P}_{e,i} ^{\text{MAP}} (N,R)$ and
is optimal. Compare this to the suboptimal bit decoder which first decodes the whole block
and then extracts the $i$-th bit. The probability of error associated to this decoder is at
most $\mathbb{P}_{e} ^{\text{MAP}} (N,R)$ since any time the block is decoded correctly
also the $i$-th bit is decoded correctly. Therefore, for any $i$,
$\mathbb{P}_{e,i} ^{\text{MAP}} (N,R) \leq \mathbb{P}_{e} ^{\text{MAP}}(N,R)$. Step (a) follows
by maximizing over $i$.
Step (b) is Fano's inequality.  Denote
the number of $1$s in the binary expansion of $i-1$ by $\text{wt}(i)$.
Then $\bar{W}_i$ is
\begin{equation}
\bar{W}_i= ((((W\overbrace{^+)^+)^{\cdots})^+)}^{\text{wt$(i)$ times}}.
 \end{equation}
As a result, $Z(\bar{W}_i)=(Z(W))^{2^{\text{wt}(i)}}$. Thus by using the inequality $I(\bar{W_i})^2+Z(\bar{W}_i)^2 \leq 1$ (\cite{Ari09}), we have $H(\bar{W}_i) \geq \frac 12 (Z(W))^{1+2^{\text{wt}.(i)}}  $. As a result,
\begin{align*}
\mathbb{P}_e ^{\text{MAP}} (N,R) & \geq \text{max}_{i \in \mathcal{I} } \{ H(\bar{W}_i) \}\\
& \geq \text{max}_{i \in \mathcal{I} } \{ \frac 12 (Z(\bar{W}_i))^{2^{1+\text{wt}(i)}} \}.
\end{align*}
Since, $\vert \mathcal{I} \vert = 2^nR$, the set $\mathcal{I}$ must
contain an index $i$ so that $\text{wt}(i) \leq E(n,R)$. Therefore,
\begin{small}
\begin{equation*}
\mathbb{P}_C ^{\text{MAP}} (n,R) \geq \frac12
(Z(W))^{2^{1+E(n,R)}}=2^{-2^{E(n,R)+ 1+\log(-\log(Z(W)))}-1}.
\end{equation*}
\end{small}
For the specific case of polar codes, we argue as follows: Let $n \in \naturals$, $m=\log n$. Also let $0<\epsilon<1$ be a constant. Using \eqref{bound5} we obtain
\begin{align*}
& \lim_{n \to \infty}\mathbb{P}(Z_{n+m} \geq (Z(W))^{2^{m}2^{E(n,x-\epsilon)}} \mid Z_{m} \leq 2 {\rho}^m) \\ \nonumber
& \geq 1-x+\epsilon.  \nonumber
\end{align*}
Also, using Lemma~\ref{main_lemma} we get
\begin{align*}
\lim_{n\to \infty} \mathbb{P}(Z_{n+m}\leq 2^{-2^{\sum_{i=m+1}^{n+m}B_i}}\mid Z_m \leq 2 {\rho}^m)=1.
\end{align*}   
As a result of the above two inequalities we have
\begin{small}
\begin{align*}
&\lim_{n\to \infty}  \! \! \!  \mathbb{P}((Z(W))^{2^{m}2^{E(n,x-\epsilon)}}\leq Z_{n+m} \! \! \! \leq 2^{-2^{\sum_{i=m+1}^{n+m}B_i}}\mid Z_m \leq 2 {\rho}^m)\\
&\geq 1-x+\epsilon.
\end{align*}
\end{small}
Also, using the result of Theorem~\ref{main_result} part (a), it is easy to see that
 \begin{align*} \nonumber
 \lim _{n \to \infty} \mathbb{P}(Z_{n+m} \leq 2^{ - 2^{E(n+m,x)+\Theta(m)}}\mid Z_m \leq 2 {\rho}^m) =x.
\end{align*}
As a result, given that $Z_m \leq 2 {\rho}^m$, as $n \to \infty$, the following two events have non-empty intersection
\begin{small}
\begin{align*}
&A_n=\{Z_{n+m} \leq 2^{ - 2^{E(n+m,x)+\Theta(m)}}\mid Z_m \leq 2 {\rho}^m\}\\
&B_n=\{(Z(W))^{2^{m}2^{E(n,x-\epsilon)}}\leq Z_{n+m}\leq 2^{-2^{\sum_{i=m+1}^{n+m}B_i}}\mid Z_m \leq 2 {\rho}^m\}.
\end{align*}
\end{small}
But the set $A_n$ exactly represents the set of indices of the sub-channels needed in order to achieve rate $R$. Also, for every $(b_1, \cdots,b_{m+n}) \in B_n $ we have
\begin{align*}
(Z(W))^{2^{m}2^{E(n,x-\epsilon)}} \leq 2^{-2^{\sum_{i=m+1}^{n+m}b_i}}.
\end{align*}
Or by applying the function $\log(-\log())$ to both sides we obtain
\begin{align*}
\sum_{i=m+1}^{n+m}b_i   \leq m + E(n+m,x-\epsilon) + \log(-\log(Z(W))).
\end{align*}
As a result,
\begin{align*}
\sum_{i=1}^{n+m}b_i \leq E(m+n,x-\epsilon)+ \Theta(m).
\end{align*}
Now since the intersection of $A_n$ and $B_n$ is non-empty for large $n$, there exists a $(b_1,\cdots,b_{n+m}) \in A_n$ with
 \begin{align*}
\sum_{i=1}^{n+m}b_i \leq E(m+n,x-\epsilon)+ \Theta(m).
\end{align*}
And by letting $\epsilon \to 0$ and noting that $\sum_{i=1}^{n+m}b_i$ is a weight of some sub-channel, we get the result.
\section{Appendix}
\subsection{Proof of Lemma~\ref{gen_bound}}
In order to prove Lemma~\ref{gen_bound},  we first need to state the following two lemmas and afterwards we give a proof of Lemma~\ref{gen_bound}.
\begin{lemma}\label{Q}
Let $Z_n$ be a process defined by $Z_0=z_0 \in [0,1]$ and
 \begin{equation*}
 Z_{n+1}  \left\{
\begin{array}{lr}
={Z_n}^2 &  ; \text{if $B_n=1$},\\
\in [Z_n \sqrt{2-{Z_n}^2} , 2Z_n-{Z_n}^2] &  ; \text{if $B_n=0$}.
\end{array} \right.
\end{equation*}
Let $Q_n=Z_n(1-Z_n)$. Then
\begin{equation*}
 \mathbb{E}[{Q_n}^{\frac 12} ] \leq \frac 12 (\frac{1.85}{2})^n.
\end{equation*}
\end{lemma}
\begin{proof}
We have
 \begin{equation*}
 Q_{n+1}= Q_n . \left\{
\begin{array}{lr}
=Z_n(1+Z_n) &  ; \text{if $B_n=1$},\\
\in [\frac{Z_n \sqrt{2-{Z_n}^2}}{Z_n(1-Z_n)} , \frac{2Z_n-{Z_n}^2}{Z_n(1-Z_n)}] &  ; \text{if $B_n=0$}.
\end{array} \right.
\end{equation*}
As a result
\begin{small}
\begin{align*}
 & \mathbb{E}[{Q_{n+1} }^{\frac 12} \mid Q_n] \\
 & \leq \frac {Q_n^{\frac 12}}{2} [\max_{Z_n\sqrt{2-Z_n ^2}\leq x \leq Z_n(2-Z_n)} \{ \sqrt{\frac{x(1-x)}{Z_n(1-Z_n)}}\}+ \sqrt{Z_n(1+Z_n)}]\\
 & \leq \frac {Q_n ^ {\frac 12}}{2} [\max_{z\sqrt{2-z ^2}\leq x \leq z(2-z),0 \leq z \leq 1} \{\sqrt{\frac{x(1-x)}{z(1-z)}}+ \sqrt{z(1+z)}\}]\\
& \leq Q_n ^{\frac 12} \frac{1.85}{2}.
\end{align*}
\end{small}
Therefore,
\begin{equation*}
\mathbb{E}[Q_n ^{\frac 12}] \leq (\frac{1.85}{2})^n \mathbb{E}[Q_0 ^ {\frac 12}] \leq \frac 12(\frac{1.85}{2})^n  .
\end{equation*}
\end{proof}
 \begin{lemma}\label{lower_gen}
 Let $E_n$ be the process defined by $E_0 = e_0 $ and
\begin{equation} \label{E}
E_{n+1}=  \left\{
\begin{array}{lr}
{E_n}^ 2   &  ; \text{if } B_n=1, \\
E_n\sqrt{2-{E_n}^2} &  ; \text{if } B_n=0.
\end{array} \right.
\end{equation}
For $n \in \naturals$ we have:
\begin{equation*}
 \mathbb{P} (E_n \geq 1-  2^{-2^{\sum_{i=1} ^{n} \bar{B}_i}} ) \geq 1-2\sqrt{2} \sqrt{1- e_0 ^2}.
\end{equation*}
\end{lemma}
\begin{proof}
We have:
\begin{equation*} \left\{
\begin{array}{lr}
1-E_{n+1} ^2 =1-E_{n} ^4 \leq 2(1-E_{n} ^2 )    &  ; \text{if } \bar{B}_n=1,\\
1-E_{n+1} ^2 = (1-E_{n} ^2)^2  &  ; \text{if } \bar{B}_n=0.
\end{array} \right.
\end{equation*}
Hence the process $\bar{E}_n=1- E_n ^2$ with the initial condition $\bar{E}_0=1- e_0 ^2$ is stochastically dominated by the process $Z_n ^{u}$ given by \eqref{upper_process} and $z_0 ^{u} = 1- e_0 ^2$. Therefore, by \eqref{eq:behavior} we have:
\begin{equation*}
 \mathbb{P} (\bar{E}_n \leq  2^{- 2^{\sum_{i=1} ^{n} \bar{B}_i}}) \geq 1-2\sqrt{2} \sqrt{1- e_0 ^2}.
\end{equation*}
Also,
\begin{align*}
&  \mathbb{P} (\bar{E}_n \leq  2^{-2^{\sum_{i=1} ^{n} \bar{B}_i}})\\
& =  \mathbb{P} (E_n ^2 \geq 1- 2^{-2^{\sum_{i=1} ^{n} \bar{B}_i}})\\
& =  \mathbb{P} (E_n  \geq (1- 2^{-2^{\sum_{i=1} ^{n} \bar{B}_i}})^{\frac 12})\\
& \leq  \mathbb{P} (E_n  \geq 1- 2^{- 2^{\sum_{i=1} ^{n} \bar{B}_i}}).
\end{align*}
As a result, we have
\begin{equation*}
 \mathbb{P} (E_n \geq 1-  2^{-2^{\sum_{i=1} ^{n} \bar{B}_i}} ) \geq 1-2\sqrt{2} \sqrt{1- e_0 ^2}.
\end{equation*}
\end{proof}
Using the above two lemmas, we now prove Lemma~\ref{gen_bound}. Let $\rho_1= (\frac{1.85}{2 \times \rho})^2$. Consider the process $Q_n =Z_n (1- Z_n)$. According to Lemma~\ref{Q} and by using the Markov inequality
\begin{align*}
\mathbb{P}(Q_n \geq { \rho_1 }^n) = \mathbb{P}({Q_n}^{ \frac 12} \geq  (\rho_1)^{\frac{n}{2}}) \leq ({\frac{1.85}{2 \sqrt{\rho_1}}})^n={\rho}^n.
\end{align*}
As a result,
\begin{align*}
 &  \mathbb{P}( \frac{1- \sqrt{1-4 {\rho_1}^n}}{2} \leq Z_n \leq \frac{1+\sqrt{1-4 {\rho_1}^n}}{2}) \\
& = \mathbb{P}(Q_n \leq {\rho_1}^n)  \leq {\rho}^n.
\end{align*}
Consider a partitioning of the interval $[0,1]$ into the three intervals
\begin{small}
\begin{align*}
[0,1]=[0, \frac{1- \sqrt{1-4 {\rho_1}^n}}{2}] & \cup [\frac{1- \sqrt{1-4 {\rho_1}^n}}{2},\frac{1+ \sqrt{1-4 {\rho_1}^n}}{2}] \\
& \cup [\frac{1+ \sqrt{1-4 {\rho_1}^n}}{2},1],
\end{align*}
\end{small}
and define $A,B$ and $C$ as
\begin{align*}
& A= \mathbb{P}(Z_n \leq \frac{1- \sqrt{1-4 {\rho_1}^n}}{2}),\\
& B=\mathbb{P}(\frac{1- \sqrt{1-4 {\rho_1}^n}}{2} \leq Z_n \leq \frac{1+ \sqrt{1-4 {\rho_1}^n}}{2}),\\
&C= \mathbb{P}(Z_n \geq \frac{1+ \sqrt{1-4 {\rho_1}^n}}{2}).
\end{align*}
Also let $A'$, $B'$ and $C'$ be the fraction of $A,B$ and $C$ respectively that will eventually (as $n \to \infty$) go to zero. Clearly we must have
\begin{equation}
A'+B'+C'= \mathbb{P}(Z_{\infty}=0)= I(W).
\end{equation}
Clearly $B' \leq B \leq \rho ^n $. To upper-bound $C'$ note that if we consider the process $E_n$  given by \eqref{E} and $E_0=e_0= \frac{1+ \sqrt{1-4 {\rho_1}^n}}{2} $ then by \eqref{general_process} it is easy to see that $\mathbb{P}(E_{\infty}=0)$ is an upper bound on $C'$.  Thus we have
\begin{align*}
C' & \leq \mathbb{P}(E_{\infty}=0)\\
& \leq 4 \sqrt{1- e_0 ^2}\\
& = 2\sqrt{2} \sqrt{\rho_1 ^n + \frac{1- \sqrt{1-4 {\rho_1}^n}}{2}} \\
& \leq 2\sqrt{2} \sqrt{\rho_1 ^n + \frac{1- (1-4 {\rho_1}^n)}{2}} \\
& \leq 2\sqrt{6{\rho_1}^n}.
\end{align*}
Therefore,
\begin{align*}
 \mathbb{P}(Z_n  \leq \frac{1- \sqrt{1-4 {\rho_1}^n}}{2})& = A\\
 &\geq A'\\
 &=I(W) -B'-C'\\
 &\geq I(W)-\rho ^n - 2\sqrt{6} {\rho_1}^{\frac{n}{2}}.
\end{align*}
As a result, since $\rho \geq \rho_1$ we have $\frac{1- \sqrt{1-4 {\rho_1}^n}}{2} \leq 2\rho^n$, and we get
\begin{equation*}
\mathbb{P}(Z_n \leq 2 {\rho}^n) \geq I(W) -  (1+ 2 \sqrt{6}) {\rho_1}^{\frac{n}{2}}.
\end{equation*}
Thus part (a) now follows by letting $\alpha_1= 1+ 2 \sqrt{6}$. For the proof of part (b), let $A''$, $B''$ and $C''$ be the fraction of $A,B$ and $C$ respectively that will eventually (as $n \to \infty$) go to one. Clearly we must have
\begin{equation}
A''+B''+C''= \mathbb{P}(Z_{\infty}=1)= 1- I(W).
\end{equation}
 Clearly $B'' \leq B \leq \rho ^n $. To upper-bound $A''$ note that if we consider the process $\bar{Z}_n$ given by  $\bar{Z}_0 =\frac{1- \sqrt{1-4 {\rho_1}^n}}{2}$ and
 \begin{equation*}
\bar{Z}_{n+1}=  \left\{
\begin{array}{lr}
\bar{Z}_n ^{2}  &  ; \text{if $B_i=1$},\\
2 \bar{Z}_n -\bar{Z}_n ^{2} &  ; \text{if $B_i=0$},
\end{array} \right.
\end{equation*}
  then $\mathbb{P}(\bar{Z}_{\infty} = 1)$ is an upper bound on $A''$. Therefore we have $A'' \leq \frac{1- \sqrt{1-4 {\rho_1}^n}}{2}$. As a result, $C$ can be bounded from below by
\begin{align*}
 \mathbb{P}(Z_n  \geq \frac{1+ \sqrt{1-4 {\rho_1}^n}}{2})& = C\\
 &\geq C''\\
 &=1-I(W) -A''-B''\\
& \geq 1-I(W)- {\rho} ^n-\frac{1- \sqrt{1-2 {\rho_1}^n}}{2}\\
& \geq 1-I(W)- {\rho} ^n-4 {\rho_1}^n.
\end{align*}
And since $\rho \geq \rho_1$, we get the result in a similar way as part (a) by taking $\alpha_2 =5$.
\bibliographystyle{IEEEtran}
\bibliography{lth,lthpub}
\end{document}

%% file: definition.tex
\newtheorem{theorem}{Theorem}

\newcommand{\btheo}{\begin{theorem}}
\newcommand{\etheo}{\end{theorem}}
\newcommand{\bproof}{\begin{proof}}
\newcommand{\eproof}{\end{proof}}
\newtheorem{definition}[theorem]{Definition}
\newcommand{\bdefi}{\begin{definition}}
\newcommand{\edefi}{\end{definition}}
\newtheorem{fact}[theorem]{Fact}
\newcommand{\bprop}{\begin{fact}}
\newcommand{\eprop}{\end{fact}}
\newtheorem{corollary}[theorem]{Corollary}
\newcommand{\bcor}{\begin{corollary}}
\newcommand{\ecor}{\end{corollary}}
\newtheorem{example}[theorem]{Example}
\newcommand{\bex}{\begin{example}}
\newcommand{\eex}{\end{example}}
\newtheorem{lemma}[theorem]{Lemma}
\newcommand{\blemma}{\begin{lemma}}
\newcommand{\elemma}{\end{lemma}}
\newtheorem{remark}[theorem]{Remark}
\newcommand{\bremark}{\begin{remark}}
\newcommand{\eremark}{\end{remark}}
\newtheorem{conj}[theorem]{Conjecture}
\newcommand{\bconj}{\begin{conj}}
\newcommand{\econj}{\end{conj}}

\newcommand{\naturals}{\ensuremath{\mathbb{N}}}

\def\0{{\tt 0}} 
\def\1{{\tt 1}} 
\def\?{{\tt *}} 
\renewcommand{\mid}{\,|\,}

